\documentclass[journal,twoside]{IEEEtran}

\ifCLASSINFOpdf
  \usepackage[pdftex]{graphicx}

   \DeclareGraphicsExtensions{.pdf,.jpeg,.png}
\else

\usepackage[dvips]{graphicx}
\usepackage{psfrag}

   \graphicspath{{./fig/}}

   \DeclareGraphicsExtensions{.eps}
\fi

\usepackage[cmex10]{amsmath}
\usepackage{amsfonts,mathabx}

\usepackage{algorithmic}

\usepackage{mdwmath}
\usepackage{mdwtab}

\usepackage{hyperref}
\usepackage{setspace}

\usepackage[on]{auto-pst-pdf}
\usepackage{booktabs}
\renewcommand{\arraystretch}{1.2}
\newcommand{\ra}[1]{\renewcommand{\arraystretch}{#1}}

   \usepackage{color}
   \definecolor{skel}{rgb}{1,0.5,0.5}
   \providecommand{\skel}[1]{{}}
\providecommand{\eqdef}{\overset{def}{=}}
\providecommand{\FTpair}{\overset{\mathcal F}{\longleftrightarrow}}

\providecommand{\Id}{\mathbb{I}}

\providecommand{\Cspace}{\mathbb{C}}

\providecommand{\Zspace}{\mathbb{Z}}
\providecommand{\Nspace}{\mathbb{N}}

\providecommand{\Esp}[1]{\mathbb{E}\left[#1\right]}

\providecommand{\mat}[1]{\boldsymbol{#1}}
\providecommand{\vct}[1]{\boldsymbol{#1}}

\providecommand{\w}{\omega}
\providecommand{\eq}{\:=\:}

\newtheorem{proposition}{Proposition}
\newtheorem{definition}{Definition}

\begin{document}

\title{Fast and Robust Parametric Estimation of Jointly Sparse Channels}

\author{Yann~Barbotin~\IEEEmembership{Student Member,}
        and~Martin~Vetterli,~\IEEEmembership{Fellow,~IEEE}%
\thanks{Y. Barbotin and M. Vetterli are with the faculty of Informatics and Communications at \'{E}cole Polytechnique F\'{e}d\'{e}rale de Lausanne, Switzerland}%
\thanks{This work has been submitted to the IEEE for possible publication. Copyright may be transferred without notice, after which this version may no longer be accessible.}%
\thanks{This research is supported by \emph{Qualcomm Inc.}, \emph{ERC Advanced Grant – Support for Frontier Research - SPARSAM Nr : 247006} and \emph{SNF Grant - New Sampling Methods for Processing and Communication Nr : 200021-121935/1}.}}

\markboth{\emph{submitted to} IEEE JOURNAL ON EMERGING AND SELECTED TOPICS IN CIRCUITS AND SYSTEMS}%
{BARBOTIN \& VETTERLI: Fast and Robust Estimation of Jointly Sparse Channels}

\maketitle

\begin{abstract}

 We consider the joint estimation of multipath channels
  obtained with a set of receiving antennas and uniformly probed in
  the frequency domain. This scenario fits most of the modern outdoor
  communication protocols for mobile access \cite{3gpplte} or digital
  broadcasting \cite{dvbt} among others.

Such channels verify a Sparse Common Support property (SCS) which was
used in \cite{Barbotin2012} to propose a Finite Rate of Innovation
(FRI) based sampling and estimation algorithm.
In this contribution we improve the robustness and computational complexity
aspects of this algorithm. The method is based on projection in Krylov subspaces to
improve complexity and a new criterion called the Partial Effective
Rank (PER) to estimate the level of sparsity to gain robustness.

  If $P$ antennas measure a $K$-multipath channel with $N$ uniformly sampled
  measurements per channel, the algorithm possesses an $\mathcal O(KPN\log N)$
  complexity and an $\mathcal
  O(KPN)$ memory footprint instead of  $\mathcal O(PN^{3}
 )$ and  $\mathcal O(PN^{2})$ for the direct
  implementation, making it suitable for $K\ll N$. The sparsity is
  estimated \emph{online} based on the PER, and the
  algorithm therefore has a sense of \emph{introspection} being able
  to relinquish sparsity if it is lacking.

  The estimation performances are tested on field measurements with synthetic
  AWGN, and the proposed algorithm outperforms non-sparse
  reconstruction in the medium to low SNR range ($\leq 0$dB),
  increasing the rate of successful
  symbol decodings by $1/10^{\mathrm{th}}$ in average, and
  $1/3^{\mathrm{rd}}$ in the best case. The experiments also show that
  the algorithm does not
  perform worse than a non-sparse estimation algorithm in non-sparse
  operating conditions, since it may fall-back to it if the PER
  criterion does not detect a sufficient level of sparsity.

  The algorithm is also tested against a method assuming a
  ``discrete'' sparsity model as in Compressed Sensing (CS). The conducted test indicates a trade-off
  between speed and accuracy.
  
\end{abstract}

\begin{IEEEkeywords}
Channel Estimation, Sparse, Finite Rate of innovation, Effective
Rank, Krylov Subspace.\end{IEEEkeywords}

\IEEEpeerreviewmaketitle

\section{Introduction}

\IEEEPARstart{C}{ommunications} between two parties are subject to
unknowns: noise and filtering by the Channel Impulse Response (CIR).
With respect to decoding, the noise is treated as nuisance parameters,
and the CIR coefficients as unknowns to be estimated as precisely as possible to
maximize the decoding rate. For this purpose, the channel can be used
to transmit a signal known at both ends --- the \emph{pilots} ---
to gain knowledge about the CIR.

It dictates a trade-off between the portion of the channel reserved to
the pilots --- thus lost to \emph{data} --- and the decoding error
rate due to bad channel estimation, both affecting the communication
bitrate.

The two interdependent aspects of channel estimation are therefore the selection of
pilots and the design of an estimation algorithm. We focus on the
later and use uniform DFT pilot layouts shown in
\figurename~\ref{fig:pilots}, which are found in modern communication standards
using OFDM. 
\begin{figure}[!t]
  \centering {\footnotesize
    \psfrag{f}{\tiny frequency}
    \psfrag{t}{\tiny time}
    \psfrag{a}[cB][cB]{ a) \emph{block} layout}
    \psfrag{b}[cB][cB]{ b) \emph{comb} layout}
    \psfrag{c}[cB][cB]{ c) \emph{scattered} layout}
\includegraphics[width=\columnwidth]{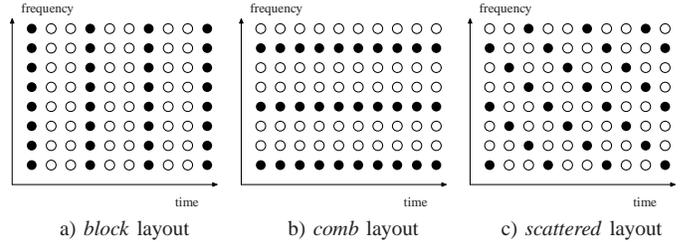}
}
\caption{In pilot assisted OFDM communications, time-frequency slots
  are reserved (black) for pilots, thus providing a sampling of the
  CIR in time and frequency.}
\label{fig:pilots}
\end{figure}
The pilots provide information about the CIR, and so does
an priori knowledge about its structure.

In the noiseless case, the CIR can be perfectly recovered with a
finite set of samples if it
perfectly obeys the a priori known structure, thus providing a \emph{sampling
  theorem} --- e.g. uniform pilots in time at the Nyquist rate
characterize uniquely bandlimited signals.

In this paper we study \emph{Sparse Common Support} (SCS) channels,
i.e. channels sharing a common structure of very
low-dimension. \figurename~\ref{fig:scschannel} shows an example of
SCS channels.

\begin{figure}[!t]
  \centering {\footnotesize 
    \psfrag{P}[c][c][1][-90]{$P$ SCS channels}
    \psfrag{K}[ct][ct]{$K$ signal components}
    \psfrag{c}[ct][ct]{$K$ scatterers}
\includegraphics[width=\columnwidth]{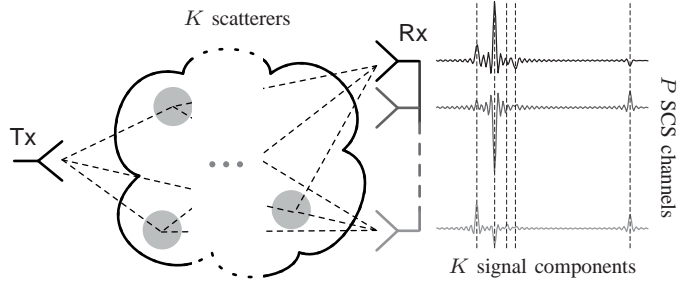}
    }
\caption{The ideal SCS channel model is a set of $P$ channels of
  bandwidth $B$ each
  having $K$ components aligned in time. Assuming complex valued
  signal components, the total number of unknowns is
  $(2P+1)K$ instead of $3PK$ for a sparse model  with independent time
  of arrivals (ToA), or
  $2P$ times the Nyquist Rate for a bandlimited model.}
\label{fig:scschannel}
\end{figure}

\subsection{Problem definition}

Imposing a structure may not lead to a trivial linear
system of equations as it is the case for Shannon/Nyquist (projection
in a linear subspace).
The SCS structure is a union of $K$ unidimensional subspaces
\cite{Lu2008,Barbotin2012} shared by $P$ channels which leads to an
estimation problem exactly and efficiently solvable in the
framework of \emph{Finite Rate of Innovation} sampling
\cite{Blu2008}. In \cite{Barbotin2012} an estimation algorithm SCS-FRI
is proposed and studied.

This leads to two questions. First, as seen in \figurename~\ref{fig:sps}.a), the SCS-FRI algorithm selects the
$K$ subspaces from an infinite and uncountable
set\footnote{notwithstanding the limitations of machine precision.} which could be
sampled. This is the approach taken by \emph{Compressed sensing} (CS) or
the \emph{sparse representation} literature. The
limiting case is to reduce the set to form a basis as shown in
\figurename~\ref{fig:sps}.b). The CIR can be perfectly and uniquely represented with
elements of this set, but it probably won't have a $K$-sparse
representation.
A possible trade-off is to enlarge the set to form a \emph{frame} as in
\figurename~\ref{fig:sps}.c). Reconstruction becomes more complex, but the
shift-invariance provided by the frame \cite{Aldroubi2001} gives a
better $K$-sparse representation.

\begin{figure*}[!t]
  \centering
  {\footnotesize
    \psfrag{A}[cB][cB]{Model \& Acquisition}
    \psfrag{L}[cB][cB]{$\mathrm{sinc}_{1/2\pi T}$}
    \psfrag{T}{$T=1/N$}
    \psfrag{r}{$t_{0}$}    
    \psfrag{q}{$t_{0}\in[0,\ 1[$}
    \psfrag{p}{$c_{0}$}    
    \psfrag{a}{a)}
    \psfrag{b}{b)}
    \psfrag{c}{c)}
    \psfrag{0}{$0$}
    \psfrag{1}{$T$}
    \psfrag{2}{$nT$}
    \psfrag{3}{$T/4$}
    \psfrag{4}{$nT/4$}
    \psfrag{t}{$t$}
    \psfrag{Z}[lB][lB]{$n<N,\in\Zspace$}
    \psfrag{z}[lB][lB]{$n<4N,\in\Zspace$}
    \psfrag{R}[cB][cB]{$t\in [0,\ 1[$}    
    \psfrag{x}[bl][l][1][90]{Synthesis functions}
    \psfrag{y}[bl][l][1][90]{1-sparse representation}
  \includegraphics[width=\textwidth]{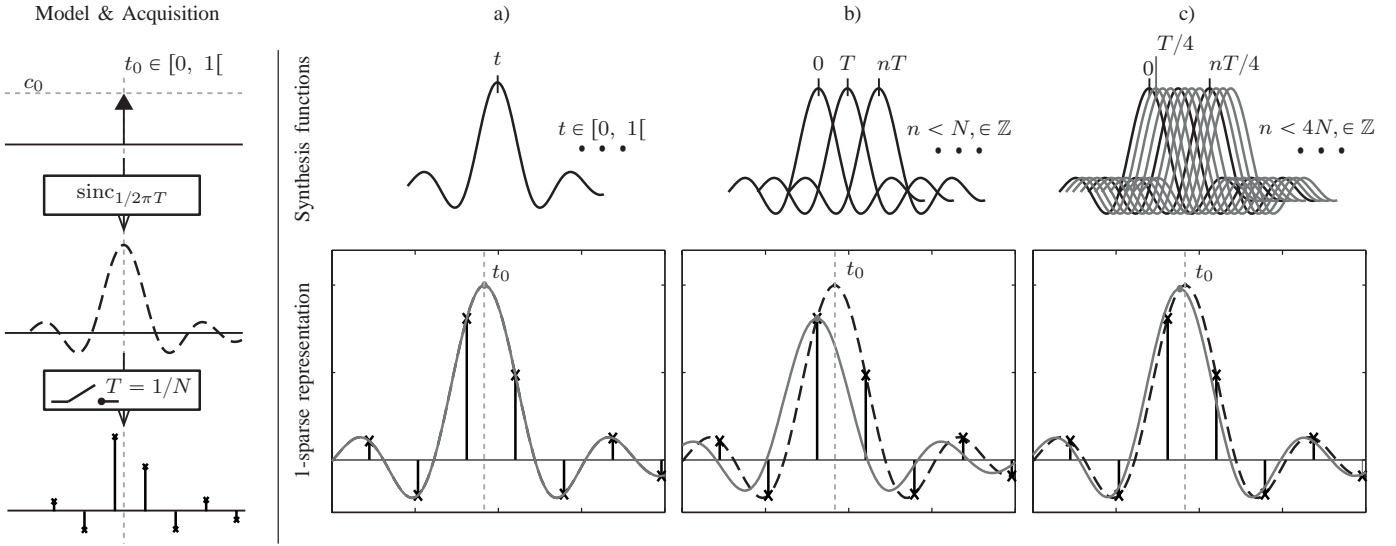}
  }
  \caption{Consider a single pulse ({\tiny$ -\
      -\ -$}, dashed curve)
    critically sampled as shown on the left column. From the samples
    ({\tiny$ -\!\!\!\times$}, stems), 1-sparse representations of the
    original signal are computed.
    In a), the problem is treated as a
    parametric estimation of $t_{0}$ and $c_{0}$ as in the FRI framework. Conceptually, the signal component is
    chosen from the infinite and uncountable set of the pulse shape and
    its shifts in $[0,\ 1[$. The original signal has a perfect
    1-sparse representation in this setup.
    In b), the signal component is chosen in a finite set of functions
    forming a basis of the signal space. The original signal can be
    represented by these functions, but it does not have a 1-sparse
    representation in general. In c), three times more synthesis
    functions are added to the set, to form a frame. The signal has a
    much closer 1-sparse representation in this frame thanks to the
    shift invariance introduced by the redundancy between the
    synthesis functions, but the estimation becomes combinatorially
    more complex. The estimation frameworks b) and c) are referred as
    \emph{discrete sparsity} which is used in \emph{Compressed
      Sensing} (CS), and the estimation is subject to a trade-off between accuracy
    and complexity.}
\label{fig:sps}
\end{figure*}

Second, models are an approximation of reality ---
e.g. \figurename~\ref{fig:scsornot} --- and it is not immediately
clear which amount of modelization error can be tolerated in practice.
This is especially important as we rely on a very specific signal structure,
and this question can only be answered by trials on field
measurements.

\subsection{Contributions}

With these considerations in mind, we adress the estimation of
\emph{Sparse Common Support} (SCS) channels from DFT-domain measurements
(pilots).
SCS channels estimation with FRI was studied in a previous paper \cite{Barbotin2012}, and
we hereby focus on computational issues and robustness.
A fast algorithm is derived based on Krylov subspace projections. Its
main advantage is to have a computation and memory cost proportional
to the sparsity level $K$. It relies on FFT evaluations for the
heavy load computations, which is particularly appealing for embedded DSP
applications.

The sparsity level is unknown in practice, and we derive a heuristic estimate
of it using a new measure  called the Partial Effective Rank
(PER). The PER tracks the ``effective dimension'' of the Krylov subspace
as its size is increased, and can therefore be estimated
online unlike other information criteria such as Akaike, MDL
\cite{Rissanen1978} or
the EDC \cite{Zhao1986}.

To assess the performances of the proposed algorithm, simulations are performed on measured CIRs to which synthetic
AWGN is added.
We compare the obtained results
with an algorithm exploiting discrete sparsity to see if the input
shift sensitivity described in \figurename~\ref{fig:sps} is a
practical issue.

\begin{figure}[!t]
  \centering {\footnotesize 
      \psfrag{T}[Bc][Bc]{\Large Tunnel}
      \psfrag{y}[c][c][1][90]{CIR [sample]}
      \psfrag{P}[Bc][Bc]{\tiny$1\cdots P$}
    \psfrag{x}[ct][ct]{Time [s]}
    \psfrag{c}[l]{\small$\underset {\log_{10}}{\mathrm{Mag.}}$}
        \psfrag{a}[ct][ct]{a) Evolution of CIR over time}
    \psfrag{b}[ct][ct]{b) From ``sparse'' to ``not so sparse'' in $20$s }
    \psfrag{o}[c][c][1][90]{$t=10$s}
    \psfrag{t}[c][c][1][90]{$t=30$s}
    \includegraphics[width=0.9\columnwidth]{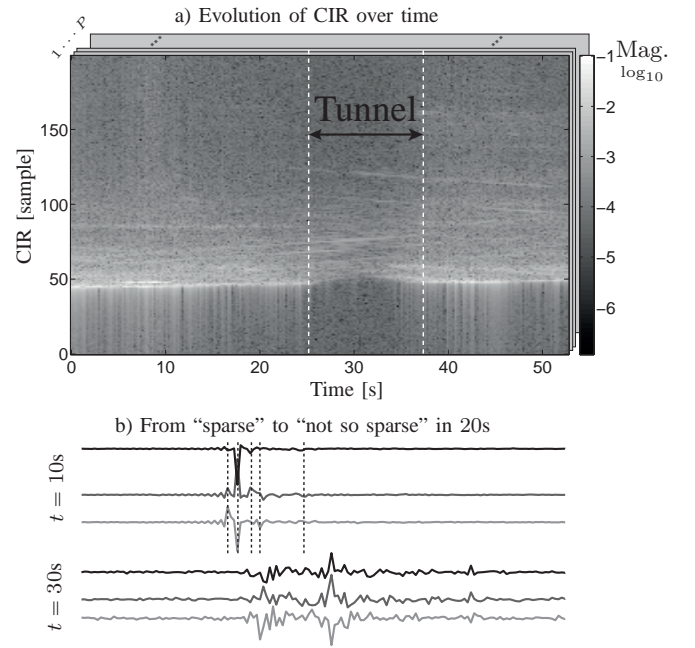}
    }
\caption{Field measurements collected in \cite{Hofstetter2002}. The
  receiver is a base-station with $P=8$ antennas, and the transmitter
  is mobile.
  The image a) shows the magnitude of the first antenna's CIR. The
  channel is qualitatively sparse except when the transmitter
  goes through a tunnel. The real part of the CIR for three different
  antennas is shown in b) confirming the common support property and
  the transient nature of sparsity.}
\label{fig:scsornot}
\end{figure}

\subsection{Outline}

First, we will explicit the SCS channel model for aerial electromagnetic
transmissions, and see under which conditions this model is
relevant. Our conclusion is that the SCS property may not always be
verified, which is confirmed by the data shown in
\figurename~\ref{fig:scsornot}. It establishes the requirements for a robust
SCS estimation algorithm aware of the
operating conditions.

Second, we quickly review the application of FRI and discrete sparse
representations to the SCS
estimation problem. Details are left out since they are already
present in the literature \cite{Berger2010,Barbotin2012,Taubock2008}.

The third part is devoted to the design of a fast and robust algorithm
for the SCS estimation problem, based on Krylov subspace projections
with $\mathcal O(KPN\log(N))$ operations, requiring $\mathcal O(KPN)$
memory, where $K$ is the number of sparse components, $P$ is the
number of channels, and $N$ is the number of collected samples per
channel. For robustness, we introduce the \emph{Partial Effective
  Rank} (PER)
derived from the work of Roy \cite{Roy2007}, which only adds a marginal
$\mathcal O(K^{2})$ cost to
be evaluated \emph{online}. The PER is used to derive a heuristic
estimate of $K$ requiring little overhead. The heuristic may fail if the channel is not
sparse enough, giving a beneficial sense of adequacy to the
algorithm. In such non-sparse cases, the algorithm can yield to a
non-sparse estimation method.

The resulting algorithm --- combining FRI sampling, the PER criterion and
Krylov subspace projection --- named FRI-PERK was implemented in MATLAB and compares favorably to FRI for $N>200$.

We conclude with the application of FRI-PERK and a CS algorithm
(RA-ORMP) to field measurements collected in \cite{Hofstetter2002}
with synthetic AWGN,
and compare them with a simple non-parametric method (spectrum
lowpass interpolation). The results show the SCS property can be
exploited in the medium to low SNR bracket (below $0$dB) to
significantly lower the Symbol Error Rate (SER).

 Compared with the non-parametric method, FRI-PERK  increases the
 proportion of correctly decoded symbols by $1/10$ in average, and
 $1/3$ at best. 

\skel{
\begin{itemize}
\item multiple electromagnetic channels 
\begin{itemize}
\item when are they sparse?
\item when have they common support? 
\item a sparse example from ftw, and a non-sparse one 
\end{itemize}
\item challenges 
\begin{itemize}
\item take advantage of sparsity and common support (solved: cite trans. comm.)
\item performances awareness: estimate sparsity, and based on this
  decide to exploit it or not $\rightarrow$ necessary for field application 
\item computational efficiency: run on low-power hardware, at high rate
\end{itemize}
\end{itemize}
\subsection{outline}
\begin{itemize}
\item intro 
\item the scs estimation problem
\begin{itemize}
\item fri solution: SCS-FRI
\item cs solution: MMV compressed sensing
\end{itemize}
\item An $\mathcal O(KPN\log N)$ FRI algorithm for jointly sparse
  signals: PERK (Partial Essential Rank / Krylov) 
\begin{itemize}
\item projection in Krylov subspace (make connection to control
  theory: very well-known!)
\item fast Krylov iteration 
\item lookahead for sparsity: the Partial Essential Rank
\end{itemize}
\item Numerical results FRI (PERK) vs CS (RemBO) 
\begin{itemize}
\item simulated data
\item field data 
\item PERK speed-up
\end{itemize}
\end{itemize}
}
\section{The physics of sparse common support channels}
\subsection{Sparsity beyond simulations}
Algorithms for the estimation of a sparse signal from noisy
measurements are quite mature \cite{Davies2010,Lee2012,Barbotin2012}. The principal hurdle for
their application to physical signals is the inadequacy between the
simple theoretical model and reality.

For example, indoor electromagnetic channels are in general not
sparse, as described in the Saleh and Valenzuela model
\cite{Saleh1987}. In this model, reflections are bundled in clusters
having an exponentially decaying energy. This dynamic holds in general
\cite{Baum2005,Salz1994}, and after
demodulation, an electromagnetic
channel impulse response $h(t)$ can be described as the superposition of
clustered reflections,
	 \begin{align}
         h(t)\ \  = \ &\  \sum_{k=1}^{K}\sum_{(A_{l},\Delta_{l})\in \mathcal C_{k} }
          c_{l} \varphi(t-t_k - \Delta_{l}),\nonumber\\
          \FTpair &\ 
           \sum_{k=1}^{K}e^{-j\w t_k}\sum_{(A_{l},\Delta_{l})\in
             \mathcal C_{k} }A_l e^{-j\w\Delta_l}\widehat\varphi(e^{j\w}),\label{eq:mpchannel}
	 \end{align}
        such that $\varphi$ is the channel mask in the time domain (after demodulation), and
         $\mathcal C_{k}$ are clusters of reflections with a delay
 $t_{k}$. Within each cluster we observe reflections shifted by
 $\Delta_{l}$ from $t_{k}$ with randomly distributed
 complex-valued\footnote{After demodulation, the equivalent baseband
   channel becomes complex-valued.}
 amplitudes $A_{l}$.

 The number of clusters $K$ is usually small, but the total number of
 reflexions is not. However, if the bandwidth $\Omega_{\varphi}$ of $\varphi$
 and the maximal intra-cluster delays $\Delta_{l}$ are small
 enough, the $0^{\text{th}}$ order approximation
\begin{equation*}
\label{eq:approx0}
e^{-j\w\Delta_l}\widehat\varphi(e^{j\w})\:\approx\:\widehat\varphi(e^{j\w}),
\end{equation*}
holds at all considered frequencies $\w\in ]-\pi,\ \pi]$. If in addition the amplitudes are
distributed identically and independently enough, one of the many
avatars of the central-limit theorem may be used to obtain a
simplified model known as the ``Multipath Rayleigh channel model'':
\begin{equation}
\label{eq:rayleigh}
h(t) \eq \sum_{k=1}^{K}C_{k} \varphi(t-t_k),\quad C_{k}\sim \mathcal
N_{\Cspace}\left(\vct 0,c_{k}^{2}\Id\right).
\end{equation}

The channel is called \emph{sparse} if and only if the expected delay-spread
$\tau$ upper-bounding $t_{K}-t_{1}$ verifies
\begin{equation*}
\label{eq:sparsedef}
K/\tau \:\ll\: \Omega_{\varphi}/2\pi,
\end{equation*}
i.e. if the rate of innovation is substantially
lower than the Nyquist rate.

The requirements for sparsity are thus two-folds:
\begin{itemize}
\item The ``girth'' of each cluster must be a fraction of the
  inverse bandwidth of the channel.
\item The density of clusters must be a fraction of the channel bandwidth
\end{itemize}
The first property hints at channels with a low or medium bandwidth,
while the second one requires long-distances propagation where both
the transmitter and receivers are in a relatively clear environment,
such as outdoor communications.

The high velocity of electromagnetic (EM) waves ensures scatterers of large dimensions, such as trees,
generate clusters of modest time spread allowing the use of
the simplified model (\ref{eq:rayleigh}) instead of
(\ref{eq:mpchannel}) for channels with a bandwidth up to $100MHz$ approximately. \figurename~\ref{fig:sptable} summarizes the
conditions necessary for the channel to be sparse. The channel shown in \figurename~\ref{fig:scsornot} has a
bandwidth of $120$MHz, and an EM wave travels $2.5$m in a time-lapse
equal to the inverse-bandwidth. Therefore, in typical noisy operating conditions, reflections
having a path difference of a fraction of $2.5$m will be unresolvable.

\begin{figure}[!t]
  \centering {\footnotesize 
    \psfrag{a}[Bc][Bc]{(a)}
    \psfrag{b}[Bc][Bc]{(b)}
    \psfrag{c}[Bc][Bc]{(c)}
    \psfrag{d}[Bc][Bc]{(d)}
    \psfrag{h}[c][c][1][90]{High bandwidth $B$}
    \psfrag{l}[c][c][1][90]{Low bandwidth $B/8$}
    \psfrag{s}[cB][cB]{Short delay-spread $\tau/8$}
    \psfrag{L}[cB][cB]{Long delay-spread $\tau$}
    \includegraphics[width=0.9\columnwidth]{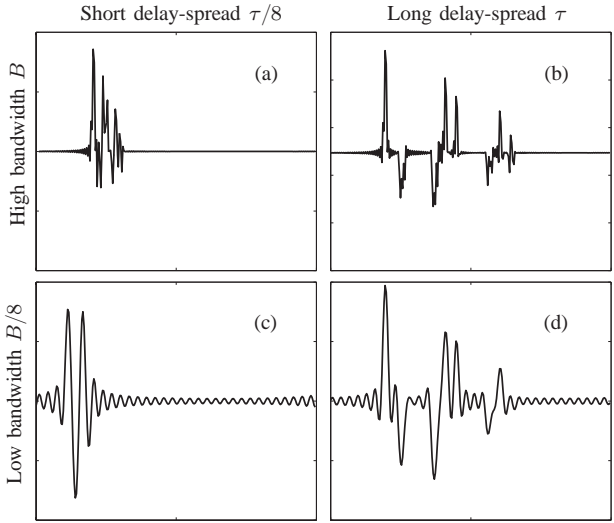}
    }
\caption{This figure shows how sparsity relates to the channel
  bandwidth and its delay-spread. All four panels (a)--(d) have the
  same number of signal components --- $80$ of them grouped in $8$
  clusters with exponentially fast energy decay. Signals (a) and (c) cannot be considered
  sparse as the rate of innovation is close to/greater than the Nyquist rate on the
time-lapse corresponding to the delay-spread, and $0$ outside
it. Signal (b) is weakly sparse, the rate of innovation is for this reason also
high. In this setup the discrete sparsity approach may be
suitable. The signal (d) can be considered sparse as only the $8$ clusters will be
resolvable in the presence of noise. The rate of innovation of this approximation is much
lower than the Nyquist rate. Even though (b) and (d) have the same
rate of innovation in a strict sense, (d) can be approximated with a
signal having $1/10^{\mathrm{th}}$ the rate of innovation of (b)
thanks to its low bandwidth. This approximation motivates the use of a
model with a low rate of innovation in the low-SNR regime where the
model approximation error has less power than the noise.}
\label{fig:sptable}
\end{figure}

On the other hand the ratio $K/\tau$ will be quite large
necessitating long range outdoor transmissions for the channel to be
considered sparse.

\figurename~\ref{fig:sptable} shows the
magnitude impulse
responses of an outdoor channel with a medium bandwidth
from the FTW\footnote{Forschungszentrum Telekommunikation Wien, \texttt{\tiny\href{http://measurements.ftw.at}{http://measurements.ftw.at}}} MIMO
dataset \cite{Hofstetter2002}, and confirms this analysis.

Note that we do not rule-out the occurrence of sparsity in
ultrawide-band communications as the channel dynamics become quite
different at very short range \cite{Cassioli2002}.
\skel{
\begin{itemize}
\item channels are not sparse in general: multiple reflections. show
  indoor source image model,... 
\item reflections happen in clusters, cite: valenzuela, scm model, uwb
  models, ... 
\item however under certain circonstances sparse approximations are
  relevants: late reflexions (long impulse responses), $\frac{1}{\mathrm{bandwidth}} <<
  \frac{\mathrm{reflector girth}}{\mathrm{propagation speed}}$.
\item 
\end{itemize}
}
\subsection{Common support}

We now consider the case where a receiver (Rx) possess several
antennas, as in SIMO and MIMO communications.
Therefore, if we consider a $1$-to-$P$ communication verifying the
multipath model (\ref{eq:rayleigh}), the receiver
observes $P$ channels
\begin{equation*}
\label{eq:onetoP}
h_{p}(t)\eq \sum_{k=1}^{K}C_{k,p} \varphi(t-t_{k,p}),\quad p=1,\dots,P,
\end{equation*}
for a total of $3KP$ unknown coefficients parametrizing the channels.

As for sparsity, the high velocity of EM waves is again a crucial
factor to establish the common support property, and the difference in amplitudes
\begin{align}
\label{eq:ts}
t_{k,1}\approx t_{k,2} \approx \cdots \approx t_{k,P},\\
C_{k,1} \neq C_{k,2} \neq \cdots \neq C_{k,P}.\label{eq:dc}
\end{align}

Indeed, if $d_{\max}$ is the maximal distance between the $P$
receiving antennas, the ToA difference between antennas is
upper-bounded by $2d_{\max}/c$.The criterion for common support (\ref{eq:ts}) is
\begin{equation*}
\label{eq:2}
d_{\max}\:\ll\: \pi\frac{c}{\Omega_{\varphi}}.
\end{equation*}

In \figurename~\ref{fig:scsornot}, $d_{\max}=60\mathrm{cm}$ and
$\pi\frac{c}{\Omega_{\varphi}}=1.25\mathrm{m}$.

To assert (\ref{eq:dc}), we need to quantify the spatial correlation
of paths amplitudes
between the receiving antennas. Using \cite{Barbotin2012} (Proposition
6), and assuming
narrow scatterers, the correlation between antennas separated by a
distance $d_{m,n}$ can be approximated as
\begin{equation*}
\label{eq:corr}
\frac{\Esp{C_{k,m}C_{k^{\prime},m}^{\ast}}}{\sqrt{\Esp{\left|C_{k,m}\right|^2}\Esp{\left|C_{k^{\prime},n}\right|^2}}}
\eq \delta_{k-k^{\prime}}J_{0}\left( d_{m,n}\w_{\mathrm{carrier}}/c \right).
\end{equation*}
Using the numbers of \figurename~\ref{fig:sptable}, correlation between
antennas does not exceed $0.5$.
In most communication scenario (\ref{eq:dc}) is verified by design as
it provides spatial diversity \cite{Rappaport2002}.

In conclusion, we see that the common support property is relatively
easy to establish as it only depends on the antennas topology. This is
not the case for sparsity, in \figurename~\ref{fig:scsornot} when
the transmitter enters a tunnel, a dense train of
reflections is observed. Hence, it is reasonable to say outdoor
communications with medium bandwidth have typically ---as opposed to
``surely''--- sparse channel
impulse responses.

\skel{
\begin{itemize}
\item Follows the same logic as sparsity 
\item restricted bandwidth + high propagation speed = strong
  correlation of ToA through space 
\item idem, through time, if speed of Tx/Rx is small
\item however the strength of reflexions varies much more rapidly. 
\item support correlation through space is more reliable than through
  time: fixed Tx/Rx structure compared to unknown Tx/Rx mobility, terrain,...
\end{itemize}
}
\section{Estimation of sparse common-support channels from DFT pilots}

Assume each frame of period $\tau_{f}$ is uniformly sampled in time
\begin{equation*}
\label{eq:26}
x_{p}[n]\eq x_{p}(n\tau_{f}/N_{f})\ ,\quad n\eq 0,\ \dots,\ N_{f}-1.
\end{equation*}

The pilot layouts in \figurename~\ref{fig:pilots} lead to estimate a set of sparse
common support channels from a uniform subset of its DFT
coefficients. 

To make notation less cumbersome, we assume without loss of generality
that $N_{f}=2MD+1$ , and use negative indices which shall be
understood ``modulo the index limit'' in general. Also the pilot index
is centered on the DC carrier
$$\mathcal P\eq \left\{ -MD,\dots,\ -2D,\ -D,\
0,\ D,\dots, \ MD\right\} $$
and is supposed to take the value 1. The total number of pilots is $N=2M+1$.

Each frame is periodically padded to ensure the
  convolution with the CIR is circular, thus
\begin{align*}
\label{eq:spl}
\widehat{\vct x}_{p}\eq& \mathrm{diag}\left( \vct 1_{\mathcal P}
\right) \widehat{\vct h}_{p},\\
\Leftrightarrow\: \vct x_{p}\eq & \mat W\mathrm{diag}\left( \vct 1_{\mathcal P}
\right)\mat W^{H}\vct h_{p},
\end{align*}

 such that $\mat W$ is the DFT matrix, $\vct 1_{\mathcal P}$ is the
 indicator of $\mathcal P$, and
\begin{equation}
 \label{eq:spsdft}
 \widehat{h}_{p}[n]\eq \sum_{k=1}^KC_{k,p}W_{N_{f}}^{nt_{k}/\tau_{s}}.
 \end{equation}
 is the sampled SCS channel model in the DFT domain.

The operator $ \mat W\mathrm{diag}\left( \vct 1_{\mathcal P}
\right)\mat W^{H}$ is the orthogonal projection in the subspace spanned
by the basis vectors of the DFT corresponding to the pilot
indices.%

With a good synchronization, the ToAs are contained in $]-\tau/2,\
\tau/2]$, and if
\begin{equation}
\label{eq:4}
D\:\leq\: \left\lfloor \frac{\tau_{s}}{\tau} \right\rfloor\ ,\quad D\in\Nspace,
\end{equation}

the original spectrum $\widehat{\vct h}_{p}$ can be recovered by
ideal lowpass interpolation

\begin{equation}
\label{eq:6}
\mat W^{H}\vct h_{p}\eq \mat C_D \mathrm{diag}\left( \vct 1_{\mathcal P}
\right)\mat
W^H \vct x_p\ ,
\end{equation}
such that $\mat C_{D}$ is circulant with entries
\begin{equation*}
\label{eq:7}
\left[\mat C_{D}\right]_{m,n}\eq \frac{\sin(\pi(m-n)/D)}{\sin(\pi(m-n)/N_{f})}\ .
\end{equation*}
This principle is at the core of OFDM communications which interleave
data and pilots in the DFT domain.

If the measurements $\vct x_p$ are corrupted by AWGN, this technique projects orthogonally the
measurements in the signal subspace and therefore minimizes the
estimation MSE if common-support property or sparsity are ignored.

\subsection{FRI approach}\label{sec:fri}

The channels in (\ref{eq:spsdft}) do not lie in linear subspaces, but
in a common union of subspace. Therefore the channels can be estimated
in two-steps

\begin{enumerate}
\item Identify jointly the $K$ subspaces (the common support). 
\item Compute the orthogonal projection of the measurements
  in the union of subspaces separately for each channel. 
\end{enumerate}
Algorithms and analysis for this problem are found in
\cite{Barbotin2012}, and are a simple extension of well-known
linear array-processing techniques such as ESPRIT \cite{Roy1989} or the
annihilating filter \cite{Tufts1982} to common-support channels.

The union of subspaces is identified by studying the column-space of
the following data matrices
\begin{align*}
\label{eq:8}
\mat
T_{p}\eq&
\begin{bmatrix}
\widehat{h}_p[0]&\widehat{h}_p[-D]&\widehat{h}_p[-2D]&\cdots\\
\widehat{h}_p[D]&\widehat{h}_p[0]&\widehat{h}_p[-D]&\cdots\\
\widehat{h}_p[2D]&\widehat{h}_p[D]&\widehat{h}_p[0]&\cdots\\
\vdots&\vdots&\vdots&\ddots
\end{bmatrix},
\end{align*}

of dimensions $(M+1)\times (M+1)$

These data matrices have a Vandermonde decomposition
\begin{align*}
\mat T_p \eq \mat J \mathcal V \mathrm{diag} \left( C_{p,1},\ \dots,\ C_{p,K}
\right)\mathcal V^{H},\\
\text{such that } \mathcal V\eq
\begin{bmatrix}
  1 & \cdots & 1\\
  W^{Dt_1/\tau_s} & \cdots & W^{Dt_K/\tau_s}\\
W^{2Dt_1/\tau_s} & \cdots & W^{2Dt_K/\tau_s}\\
\vdots&\vdots&\vdots
\end{bmatrix}\ ,
\end{align*}
is an $(M+1)\times K$ Vandermonde matrix and $\mat J$ is the exchange matrix.
This decomposition has the merit to clearly show $\mat T_{p}$ has a
column-space of dimension $K$ which depends only on the support, and
is thus the same for each channel.
Moreover any basis for this column-space verifies a
\emph{rotation invariance} property. Indeed, let
$$\mathcal
V^{\uparrow}\eqdef\left[\mathcal V\right]_{1:(M-1),:}\quad\text{and}\quad\mathcal
V^{\downarrow}\eqdef\left[\mathcal V\right]_{2:M,:}.$$
Then
\begin{equation*}
\label{eq:10}
\mathcal V^{\uparrow}\eq \mathcal V^{\downarrow}\mat \Psi\ ,\quad \mat
\Psi=\mathrm{diag}\left( W^{Dt_1/\tau_s},\ \dots,\ W^{Dt_K/\tau_s} \right).
\end{equation*}
Any basis $\mat V$ having the same span as $\mathcal V$ can be written as
$\mat V\eqdef\mathcal V\mat A$, where $\mat A$ is a fullrank $K\times
K$ matrix, therefore
\begin{align*}
\label{eq:11}
\mat V^{\uparrow}\eq& \mathcal V^{\uparrow}\mat A,\\
\eq&\mathcal V^{\downarrow}\mat\Psi\mat A,\\
\eq& \underbrace{\mathcal V^{\downarrow}\mat A}_{\mat
  V^{\downarrow}}\underbrace{\mat A^{-1}\mat\Psi\mat A}_{\mat X}.
\end{align*}
which means the support is recovered from any basis $\mat V$ of the
column-space of $\mat T_{p}$ as the phase of the eigenvalues of the solution to
\begin{equation*}
\label{eq:12}
\mat V^{\uparrow}\eq \mat V^{\downarrow}\mat X\ ,
\end{equation*}
which is the ESPRIT algorithm \cite{Roy1989}.
In the presence of AWGN, the basis $\mat V$ may be obtained in a robust manner from the SVD of the
stacked matrix
\begin{equation}
\label{eq:btoep}
\mat T=\begin{bmatrix}\mat T_{1}\\ \vdots\\ \mat T_{P}\end{bmatrix},
\end{equation}
but other subspace identification techniques may be used depending on
the measurements model.

To compare with the canonical estimation technique (\ref{eq:6}), this
method not only uses the limited-length of the delay-spread but also
sparsity and common-support. The solution of the estimation problem is
to be found in a smaller set of candidates.

We expect that with good
algorithms the estimation will be more resilient to noise. On the
other hand more restrictive models yield higher modelization
error. %

Note that sparsity is not used to reduce the number of pilots but to
make the estimation more robust to noise. See \cite{Barbotin2012} for
pilot reduction.

\skel{
\begin{itemize}
\item solved in (cite trans. comm. paper), not detailed here
\item restriction of SCS channels to $(P+1)K$ DoF.
\item recover support using well-known rotation invariance property of
  the signal subspace. 
\end{itemize}
}
\subsection{CS approach}
The measured DFT samples $\widehat{\vct x}_{p}$ corresponding to the
set of pilot subcarriers
$\mathcal P$ are
linked to the channel impulse response coefficients by
\begin{align}
\label{eq:cs}
&\mat X\eq \left[ \mat W \right]_{\mathcal P,\mathcal C}
\mat H,\\
&\mat X\eqdef \left[ [\widehat{\vct x}_{1}]_{\mathcal P},\ \dots,\   [\widehat{\vct
    x}_{P}]_{\mathcal P}\right]\ ,\ \ \mat H\eqdef \left[ {\vct h}_{1},\ \dots,\   {\vct
    h}_{P}\right],\nonumber
\end{align}
where $\mathcal C$ is the index set on which the channel impulse
responses are supported a priori. $\mathcal C$ contains contiguous
indices covering a time-lapse less than or equal to the delay-spread.
The matrix of channel coefficients $\mat H$ is assumed to be
\emph{jointly row-sparse}, i.e. only a few rows of $\mat H$ are not
null.

This is known as the MMV (Multiple Measurement Vectors) problem which is not a trivial extension of
the SMV (Single Measurement Vector) problem. It has received a lot of attention in the past few
years and good, efficient algorithms exist such as RA-ORMP
\cite{Davies2010}, Lee \& Bressler \cite{Lee2012},\dots

First notice that the channel model described in (\ref{eq:spsdft}) is
not jointly row-sparse nor sparse, unless the ToAs coincide with the
sampling points.

Second, since (\ref{eq:4}) holds, the system (\ref{eq:cs}) is invertible, which is typically
not the case in compressed sensing. In this study sparsity is used as a
regularization technique to gain robustness to noise as in Section~\ref{sec:fri}.
One could select a random subset of pilots to effectively compress the
channel measurements leaving additional space for data.

The compressed sensing literature on sparse channel estimation has
followed two tracks
\begin{enumerate}
\item\label{item:3} The ToA coincide with the sampling grid.
  \item The ToA do not coincide with a regular grid resulting in ``approximately sparse'' channels
    \cite{Taubock2008,Berger2010,Soltanolkotabi2009}. It introduces
    a model mismatch which can be mitigated using frames (increased estimation complexity).
  \end{enumerate}

  The first assumption is remote from reality and the reported
  performances shall be taken with a grain of salt, nevertheless it is
  useful to benchmark algorithms. The second one is
  more relevant and its limitations due to
  model mismatch may fade away when applied to real-world data where
  it is unavoidable anyway.

\skel{
\begin{itemize}
\item known as the MMV problem. 
\item exploits sparsity by proxy, subject to contradiction: sparse in
  time $\Rightarrow$ high bandwidth (sharp localized pulses) $\Rightarrow$ limited
  clustering. low bandwidth $\Rightarrow$ ``fatter'' pulses.
\item Many works available: mixed-norms, rembo, ... 
\item will use rembo for comparison (with a slight modification
  yielding the best results) 
\item what matters most is not the algorithm (many very good of them),
  but the adequacy of the model. 
\end{itemize}
}

\section{Application of the FRI approach: PERK}

In this section, we consider the number of signal components $K$ and
the number of antennas $P$ to be small compared to the number of
pilots $N$.

The FRI based channel estimation technique outlined in
Section~\ref{sec:fri} has two shortcomings if implemented in a
straightforward manner
\begin{enumerate}
\item\label{item:cpl} Its computational complexity and memory footprint
  are respectively $\mathcal O(N^{3})$ and $\mathcal O(N^{2)}$. Both
  are contributed by the SVD decomposition used to estimate the
  column-space of $\mat T$.
\item\label{item:spr} The sparsity level $K$ is unknown.
\end{enumerate}

The complexity \ref{item:cpl}) is especially important for channel
estimation as it is a core signal processing block at the receiver's
physical layer. It is called on several times per second, and
must operate in real-time with limited power and hardware resources.
\skel{
\begin{itemize}
\item more than implementation details 
\item channel estimation is a core functional block in a communication
  scheme: must be numerically very efficient 
\item channel sparsity cannot be taken for granted: must be robust,
  and fall-back to a classical non-sparse estimation when needed
\end{itemize}}
\subsection{An $\mathcal{O}\left( KPN\log N \right)$ FRI estimation}
Naive computation of an SVD of $\mat T$ defined in \ref{eq:btoep} to obtain a $K$-dimensional subspace of
the column-space is wasteful for two reasons. First, only $K$ out of $M+1$
principal singular pairs $(\sigma_{m},\ \vct v_{m}$ are of interest. Second, $\mat T$ is well
structured --- made of Toeplitz blocks --- and
matrix factorization techniques such as QR will destroy this structure during the
factorization process, rendering it unexploitable and
requiring an explicit storage of the data matrix.

Since we are interested in the column-space of $\mat T$, we will work
on the hermitian symmetric \emph{correlation matrix}
\begin{equation*}
\label{eq:3}
\mat T^{H}\mat T\eq =\sum_{p=1}^{P} \mat T_{p}^{H}\mat T_{p},
\end{equation*}

which has eigenpairs $(\sigma_{m},\ \vct v_{m}$, $m=0,\dots,M$.

A solution to compute only the leading eigenpairs, is to project the correlation matrix in a
\emph{Krylov subspace} \cite{Parlett1998}. This is an iterative method in which
computations are performed on the original matrix, meaning the
original structure is preserved. Hence the memory footprint is
kept low and the computational complexity is similar for each
iteration.

Krylov approximants have received a lot of attention in the numerical
analysis \cite{Parlett1998,Golub1996,Saad1980,Paige1972} or control
theory \cite{Boley1997,Boley1997a} literature,
therefore we will only quickly skim through the subject.

Projection into a \emph{Krylov subspace} is done with \emph{Lanczos
  algorithm}. A comprehensive analysis was done by Xu \cite{Xu1994a}
for the estimation of covariance matrices in linear array
processing. He proposed an $\mathcal O(N^{2})$ algorithm\footnote{$K$
  is considered small and constant compared to $N$ and $P=1$.} together with
an approximate $\mathcal O(N^{2})$ estimation of the subspace
dimension $K$. Implementation of the Lanczos algorithm is quite
involved in practice and may require costly corrections at each
iteration \cite{Paige1972,Parlett1998,Golub1996}, this is why asymptotically more expensive
$\mathcal O(N^{3})$ are generally preferred unless the system is very large.

The additional structure on the original data matrix $\mat T$ allows us
to lower the complexity  from $\mathcal O(N^{2})$ to $\mathcal O(N\log
N)$, making it appealing even for matrices of modest size, and we will derive a novel criterion to estimate the signal
subspace dimension $K$ which requires $\mathcal O(K^{2})$ computations to
be run along the subspace estimation process.

A $K$-dimensional Krylov subspace $\mathcal K$ of a $M$-dimensional hermitian matrix $\mat A$ has is
\begin{equation*}
\label{eq:kbasis}
\mathcal{K}_{K,\vct f}(\mat A)\eq \mathrm{span}_{k=1,\dots,K} \mat A^{k}\vct f ,
\end{equation*}
where $\vct f$ is an initial vector which can be randomly chosen. Note that all the
properties of Krylov subspaces only hold in probability, an
``unlucky'' initial draw may compromise them.

\begin{table*}[t]\centering  \ra{1.3}\caption{Approximative ``$\mathcal
    O$''  complexity
    for subspace identification.}
 \footnotesize{\begin{tabular}{@{}lrrrr@{}}\ \\\toprule
\textbf{Algorithm} & \textbf{Main computation}& \textbf{Storage}
&\textbf{Latency}&\textbf{Processing units (pu)} \\\midrule
 Krylov + PER& $KPN\log N$ & $K(N+1)$ &
 $KN$ \cite{Yeh2003} & $P\:\times$ FFT engines ($N+1$ points)\\
Full SVD (serial)  &$PN^{3}$& $PN^{2}$ & $PN^{3}$& 1
multipurpose pu.\\%
Full SVD (systolic array) \cite{Brent1985,Cavallaro1988}  &$
PN^{3}$& $PN^{2}$ & $N(\log N +P)$ & $N^{2}\times\:$ 2-by-2  SVD pu.\\%
\bottomrule
\end{tabular}\vspace{+2mm}\\
\parbox{0.75\textwidth}{The full SVD is done with Jacobi
    rotations and can be massively parallelized using the systolic
    array method of Brent, Luk \& Van Loan
    \cite{Brent1985}. Parallelism greatly reduces the latency of the %
    system, but since it does not reduces the number of computations
    it comes at the cost of using multiple processing units.}
} \label{tabl:cpl}\vspace{-4mm}
\end{table*}

The $k^{\mathrm{th}}$ basis vector $\mat A^{k}\vct f$ is a monomial of $\mat
A$ of  degree $k$, therefore using a three terms linear recursion on
this sequence of monomials one
can derive a sequence of orthogonal polynomials \cite{Szego1939}. This is equivalent to
say that an orthonormal basis $\mat Q_{K}$ of $\mathcal{K}_{K,\vct f}(\mat A)$ is computed by
orthogonalization of each  of $\mat A^{k}\vct f$ only against the two previous
ones and normalization. So, the main cost of the procedure is the
computation of the non-orthogonal basis vectors, which is done by
recursive  matrix-vector multiplications.

The three terms recursion used to orthogonalize the Krylov basis
implies that $\mathcal{K}_{K,\vct f}(\mat A)$ has a
tridiagonal decomposition

\begin{equation*}
\label{eq:16}
\mathcal{P}_{\vct f,K} \mat A = \mat Q_{K}^{H} \mat\Gamma_{K}\mat Q_{K}.
\end{equation*}
where $\mat Q_{K}$ is unitary and $ \mat \Gamma_{K}$ is symmetric and
tridiagonal (thanks to the 3-terms recursion). The eigenpairs
of $\mathcal{P}_{\vct f,K} \mat A$ are derived from this factorization
at little cost \cite{Parlett1998,Golub1996}, and they are called the
\emph{Ritz pairs}.

A remarkable property is that the Ritz pairs quickly converge to the principal eigenpairs of $\mat A$ as
$K$ grows. This quick convergence does not follow exactly the result of
Xu \cite{Xu1994a} --- since in our case the Toeplitz blocks are of
approximately square size --- but it can be explained
with a theorem of Saad \cite{Parlett1998}~(Theorem 12.4.1) which links
the rate of convergence of the Ritz pairs to the growth rate of
Chebyshev polynomials. The theorem shows that Ritz pairs converge faster to the
corresponding eigenpairs if the eigenvalues are farther apart.

Because of the Toeplitz structure of the data matrix, matrix-vector
multiplications with $\mat T^{H}\mat T$, which is the central step of
a Lanczos iteration has a cost of
$\mathcal{O}(PN\log(N))$. Indeed, 
\begin{equation*}
\label{eq:5}
\mat T^{H}\mat T\vct f\eq \sum_{p=1}^{P} \mat T_{p}^{H}\mat T_{p}\vct
f
\end{equation*}
is the sum of $P$ matrix-vector multiplications, each of
them realized as two consecutive Toeplitz maxtrix-vector
multiplications.
Square Toeplitz matrices of dimension $M+1$ can be embedded in
circulant matrices of dimension $2(M+1)=N+1$
\begin{align*}
\mat C_{p}\eqdef & \begin{bmatrix}\mat
  T_{p}&\mat{\widebar{T}}_{p}\\\mat{\widebar{T}}_{p} & \mat
  T_{p}\end{bmatrix},\\
\mat T_{p}\eq &\mathrm{toeplitz}(t_{p,-M},\dots,t_{p,0},\dots, t_{p,M}),\\
\mat T_{p}\eq &\mathrm{toeplitz}(t_{p,1},\dots,t_{p,M},0,t_{p,-M},\dots,t_{p,-1}).
\end{align*}

Circulant matrices are diagonalized by the DFT matrix, hence the
cost of a circulant matrix-vector multiplication is dominated by the
cost of $4$ FFT. 

Since
\begin{equation*}
\label{eq:15}
\begin{bmatrix}\mathbb I_{M} & \mathbb O_{M}\end{bmatrix}\mat C_{p}^{H}\begin{bmatrix}\mathbb I_{M} & \mathbb O_{M}\\\mathbb O_{M} &
  \mathbb O_{M}\end{bmatrix}\mat C_{p}^{H}\begin{bmatrix}\vct f\\ \vct
  0\end{bmatrix}\eq \mat T^{H}\mat T\vct f,
\end{equation*}

each Lanczos iteration has a cost dominated by $4P$ FFT\footnote{The
  DFT of the circulant generators can be precomputed.} of length $N+1$.
Since only the first half of the input and the first half of the output of
each FFT are non 0 or not needed, each FFT could be replaced by two FFT
of half the size.

The resulting algorithm fulfill all the requisites for a fast embedded
implementation since most of the computational load is put on the FFT
block, ubiquitous in communication systems. Table~\ref{tabl:cpl}

\skel{
\begin{itemize}
\item signal subspace identification with Krylov projection: lanczos
  algorithm 
\item well-known in control theory, but not widely used 
\item loss of orthogonality in the lanczos process and its mitigation
  makes it uncompetitive compared to more modern methods (e.g. QR
  factorization, ...) unless $N/K$ is very large, and $N$ is large. 
\item however, very competitive in our case since the data matrix has
  a toeplitz structure: each iteration preserves the structure! (in
  contrast to factorization methods) 
\item low $\mathcal O \left( KN \right)$ memory footprint 
\item for each channel, each iteration has a fast  $\mathcal O \left(
    N\log N \right)$ FFT based implementation: uses standard hardware
  functions (fft) to achieve speed-up.
\end{itemize}}

\subsection{On-line sparsity assessment}

In this section we introduce the \emph{Partial Effective Rank} (PER) a criterion to estimate the signal
subspace dimension working online with the Lanczos algorithm. Its main
advantage compared to other methods is to require little information about
the noise space. A formal study of the PER is deferred to an upcoming report.

\subsubsection{Shortcomings of traditional information criterions}
Information theoretic criteria such as Rissanen's MDL \cite{Rissanen1978}, Akaike
criterion or the EDC \cite{Yin1987} are powerful mathematical tools
which may be used to
evaluate the sparsity level $K$. They all follow a similar pattern,
which is to minimize
$$\mathsf{ITC}(\vct \sigma,K)\eq \mathcal
L(\vct\sigma,K)+K\cdot(2(M+1)-K)\cdot \mathsf P(M),$$
where $\mathcal L$ is the log-likelihood function based on $\vct
\sigma$ the singular values
of $\mat T$, and $K$ an estimate of the sparsity level. The
term $\mathsf P$ is a penalty growing at rate between $\mathcal O(1)$
and $o(N)$ depending on which one of the criteria is used.

Their performances are well understood \cite{Kaveh1987}, but the
evaluation of the likelihood function requires to compute the
product

$$\prod_{m=K}^{M}\sigma_{m}^{2}\eq \mathrm{det}(\mat T^{H}\mat
T)\left/\prod_{k=0}^{K-1}\sigma_{k}^{2} \right., $$

which has an algorithmic cost superior to the Lanczos algorithm itself.

Also, the argument used in \cite{Xu1994} to show consistency of a
partial evaluation of $\mathsf{ITC}(\vct \sigma,K)$ cannot be used in
our setup because the asymptotic distribution of the noise matrix
spectrum\footnote{``spectrum'' shall be understood as the SVD spectrum
  of the matrix
and not the noise Fourier power spectrum when matrices are concerned.} is
extremely different.
As the number of pilots
increases, the probability measure of the noise spectrum concentrates \cite{Bryc2006},
but its support is unbounded. This is a major difference
with the setup considered in \cite{Xu1994,Xu1994a} where the noise power spectrum
concentrates around the noise variance . Formally we can say
\begin{proposition}
Let $\mat T$ be an $(M+1)P\times (M+1)$ matrix composed of fixed and finite
number $P$ of stacked
Toeplitz blocks $\mat T_{p}$ of square dimensions $(M+1)\times (M+1)$ as in (\ref{eq:btoep}).
Assuming the generators of each block are sequences of white Gaussian
noise

\begin{equation*}
\label{eq:18}
\sigma_{\max}(\mat T^{H}\mat T/(M+1))\eq \mathcal O(\log M),
\end{equation*}
i.e. even with proper normalization, the spectral norm of an all-noise
matrix diverges.
\end{proposition}
\begin{proof}
The matrix norm induced by the euclidian vector norm on $\Cspace^{(M+1)P}$ is
submultiplicative, therefore
\begin{equation}
\label{eq:19}
\|\mat T_{1}\|\leq\left\|\begin{bmatrix}\mathbb I_{M+1} & \mathbb
    O\\\mathbb O & \mathbb O\end{bmatrix}\begin{bmatrix}\mat T &
    \mathbb O\end{bmatrix}\right\|\leq 1\cdot\|\mat T\|.
\end{equation}

We then use a result of Meckes \cite{Meckes2007} on the spectral norm
of square Toeplitz matrices with independently distributed subgaussian
generating coefficients 
\begin{equation}
\label{eq:20}
\|\mat T_{1}\|\sim \mathcal O\left(\sqrt{M\log M}\right).
\end{equation}

Using the triangle inequality and Meckes\'{} result
\begin{equation*}
\label{eq:21}
\|\mat T^{H}\mat T\|\leq \sum_{p}\|T_{p}^{H}T_{p}\|\sim P\mathcal
O(M\log M),
\end{equation*}
which together with (\ref{eq:19}) and (\ref{eq:20}) proves the proposition.

\end{proof}

It indicates that for pure AWGN measurements the spectrum of the matrix is
unbounded, and one cannot expect the signal and the noise spectrum to
separate nicely. Without a proper approximation, the  computation of
$\prod_{m=K+1}^{M}\sigma_{m}^{2}$ would have a cost of
$\mathcal{O}(PN^{2}\log(N))$ in general\footnote{Using a
  block-Levinson recursion to estimate the determinant $\mat T$, it maybe brought down to
  $\mathcal{O}((PN)^2)$, but it was not verified.}, driving the
complexity of the algorithm.

\subsubsection{The Partial Effective Rank}

The effective rank, is a matrix functional introduced by Roy
\cite{Roy2007} which may be seen as a ``convexification'' of the rank.
\begin{definition}\label{def:erank}
  Let $\mat A$ be a matrix with singular values
  $\vct\sigma=[\sigma_1,\ \dots,\ \sigma_{M}]^{T}$ in decreasing order, and
  singular values distribution
\begin{equation*}
\label{eq:17}
p_{m}\eq\sigma_{m}/\|\vct\sigma\|_{1}\ ,\quad m\eq 1,\ \dots,\ M.
\end{equation*}
  
The \emph{Effective Rank} of $\mat A$ is
\begin{equation*}
\label{eq:22}
\mathrm{erank}(\mat A)\eq e^{\mathcal H(p_{1},\dots,p_{M})},
\end{equation*}
where $\mathcal H$ is the entropy of the singular values distribution
\begin{equation*}
\label{eq:13}
\mathcal H(p_{1},\dots,p_{M})\eq -\sum_{m=1}^{M} p_{m}\log_{e}p_{m}.
\end{equation*}
\end{definition}

For elementary properties of the effective rank, we defer to
\cite{Roy2007}.

\begin{figure}[!t]
  \centering{\footnotesize
a)\\
{
  \psfrag{a}[rB][rB]{$9$dB}
  \psfrag{b}[rB][rB]{$4$dB}
    \psfrag{c}[rB][rB]{$-1$dB}
    \psfrag{d}[rB][rB]{$-6$dB}
    \psfrag{e}[rB][rB]{$-11$dB}
    
    \psfrag{k}[lB][lB]{true sparsity level}
    \psfrag{K}[ct][ct]{$K$}
    \psfrag{s}[rB][rB]{SNR:}
    \psfrag{P}[ct][ct][1][-90]{$\mathrm{PER}_{K}$}
    \includegraphics[width=0.85\columnwidth]{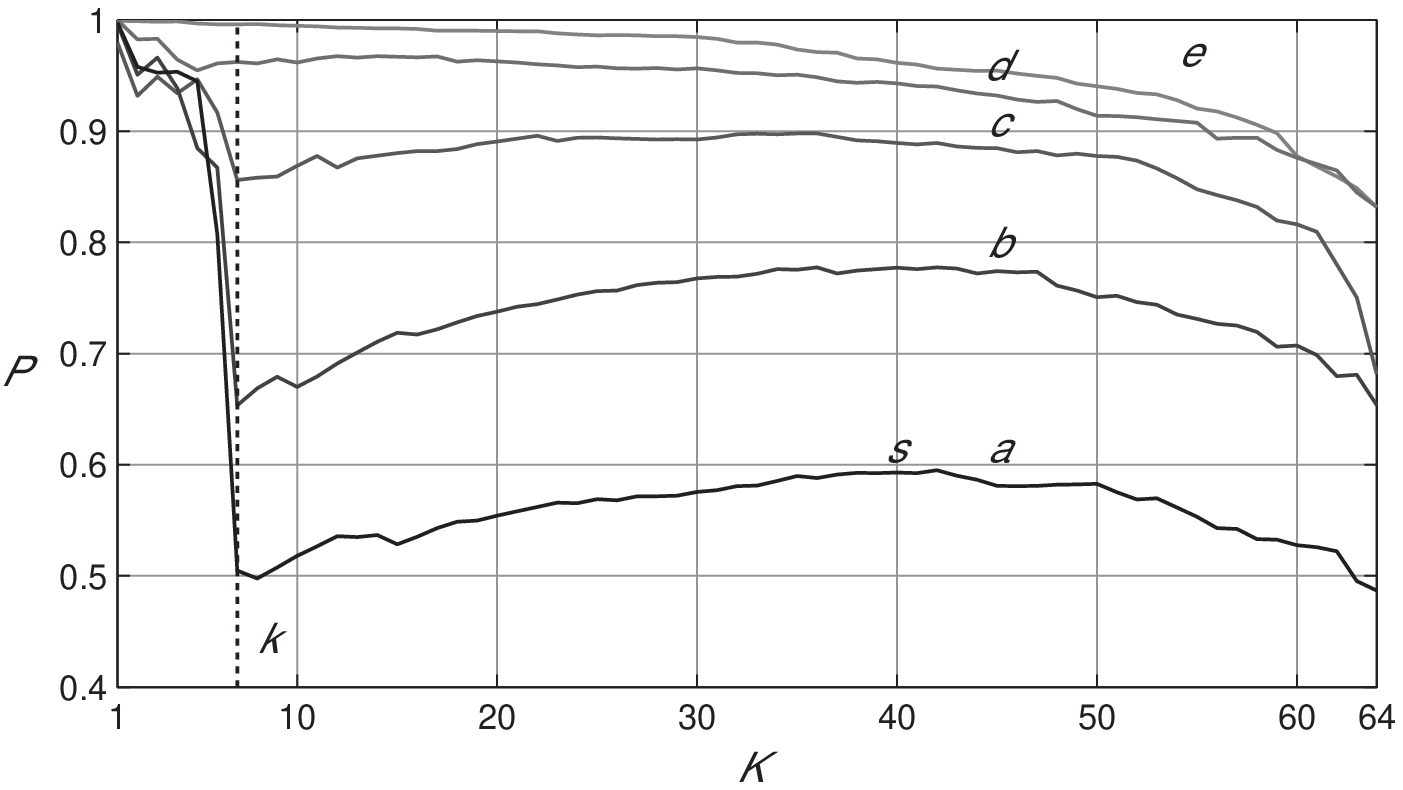}
  }
  \\\ \\ b)\\
  { 
    \psfrag{O}[cB][cB]{true value, $K=7$}
        \psfrag{k}[cB][cB]{underestimate}
    \psfrag{y}[c][c]{$K_{\mathrm{est}}$}
    \psfrag{x}[ct][ct]{SNR [dB]}
    \psfrag{c}{$\underset{\mathrm{EQ gain}}{\mathrm{dB}}$}
    \includegraphics[width=0.85\columnwidth]{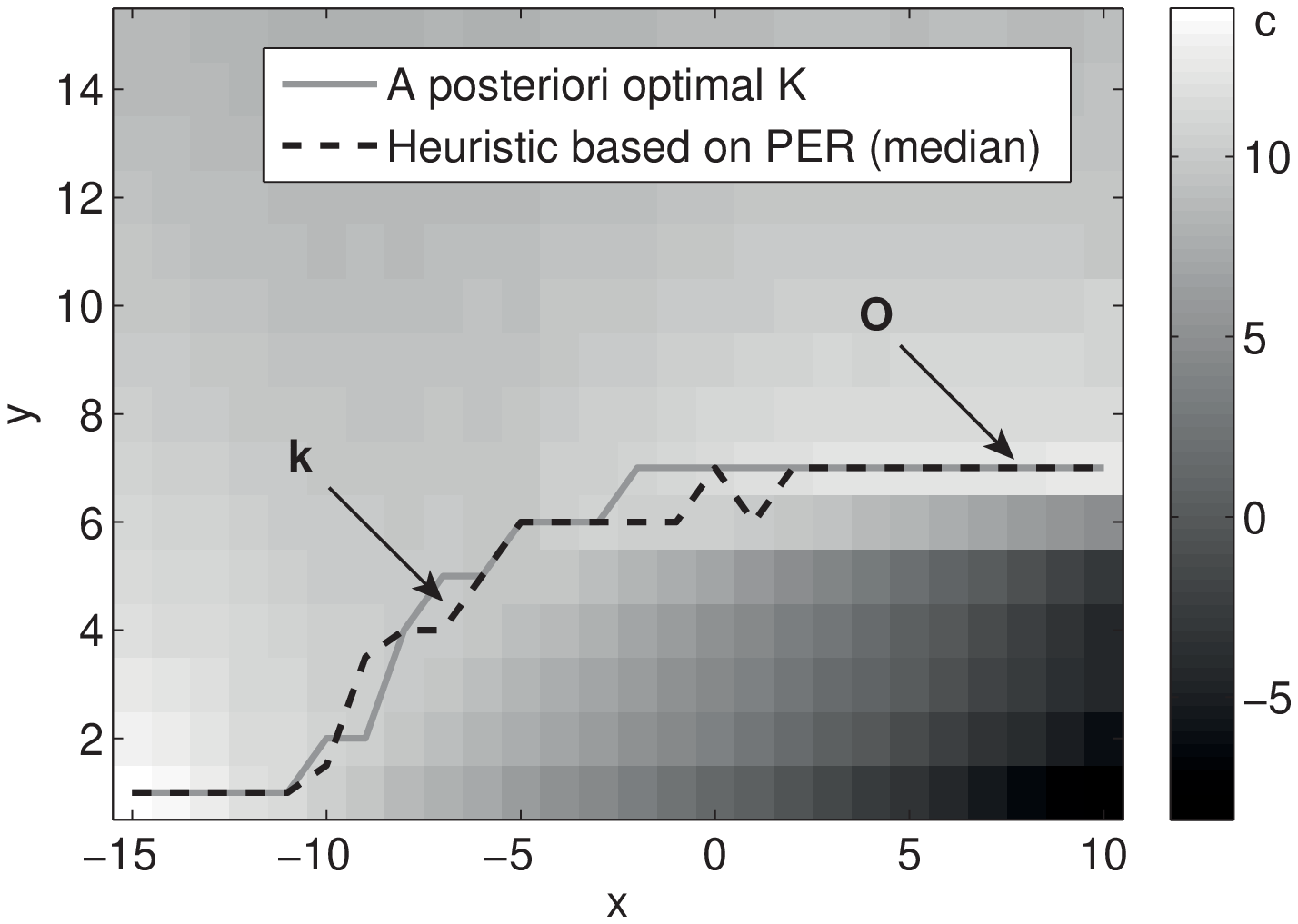}
  }
  \\\ \\ c)\\
  { 
    \psfrag{y}[c][c]{output SNR [dB]}
    \psfrag{x}[ct][ct]{input SNR [dB]}
    \includegraphics[width=0.85\columnwidth]{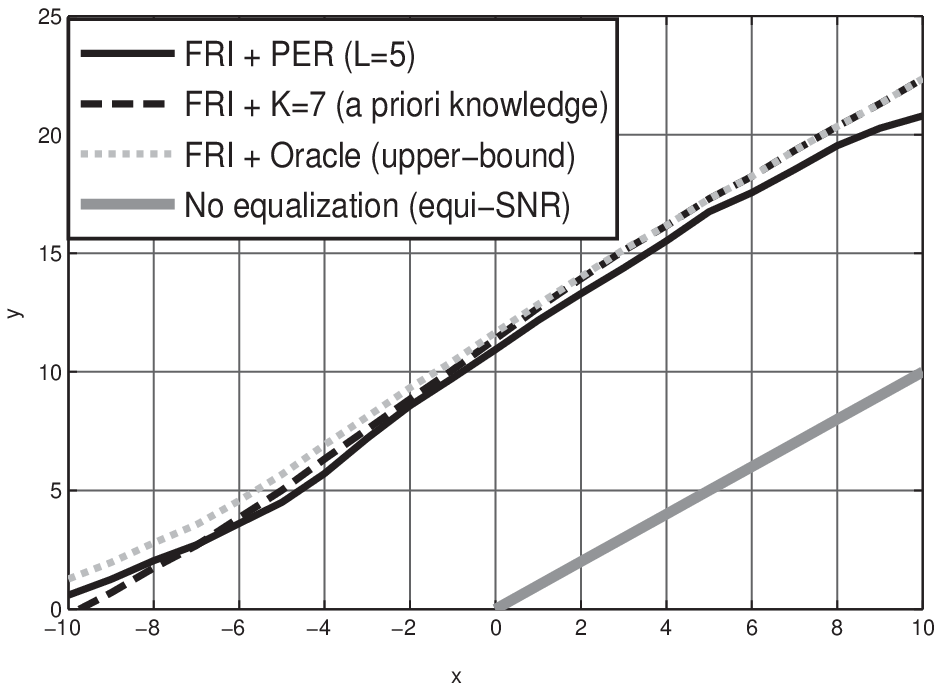}
    }}
\caption{Simulation on a signal with $7$ components. The PER curves in
(a) show a clear inflection  at $K=7$. As the SNR diminishes, the
inflection occurs at lower values of $K$ and completely disappears at
SNRs $<-10$dB. The graph b) empirically shows that the PER indicates
the number of components significantly above noise level which can be
reliably estimated: the value
of $K$ it predicts follows (in median) the one minimizing the channel
estimation error (energy of the residual). Finally, in c) the 
estimation performances obtained with the PER closely follow the one
obtained with an oracle choosing a posteriori the optimal number of
components. We see also that below $-5$dB the underestimation of $K$
by the PER is beneficial compared to the true value
$7$. At high SNR,
the PER sometimes overestimate the number of paths, mitigating the performances. }
\label{fig:perperf}
\end{figure}

We can now introduce the partial effective rank

\begin{definition}
For $\mat A$ and $\mathcal H$ as in Definition~\ref{def:erank} and 
\begin{equation*}
\label{eq:17}
p_{K,k}\eq\sigma_{k}/\|\vct\sigma_{1:K}\|_{1}\ ,\quad
k\eq 1,\ \dots,\ K\leq M,
\end{equation*}
the \emph{Partial Effective Rank} (PER) is
\begin{equation*}
\label{eq:23}
\mathrm{PER}_{K}(\mat A)\eq e^{\mathcal H(p_{K,1},\dots,p_{K,K})}.
\end{equation*}
\end{definition}

The PER verifies
\begin{proposition}
\begin{equation*}
\label{eq:24}
0\leq \mathrm{PER}_{K+1}(\mat A)-\mathrm{PER}_{K}(\mat A)\leq 1.
\end{equation*}
\end{proposition}
The lower bound $0$ is reached if and only if $\sigma_{K+1}=0$ and the
upper bound $1$ if and only if $$\sigma_{1}=\sigma_{2}=\cdots=\sigma_{K+1}.$$

The increase of the partial effective rank with $K$ reflects how
``significant'' is the $K^{\mathrm{th}}$ principal component of $\mat
A$ compared to the previous ones.

\figurename~\ref{fig:perperf} provides empirical evidences that the evolution of the
PER during the Lanczos process provides a suitable criterion to
estimate $K$.

We settle on a very simple heuristic to estimate $K$ based on the
PER. We choose the smallest $K$ such that
$$\mathrm{PER}_{K}(\mat T) - \frac{1}{L}\left(\mathrm{PER}_{K+1}(\mat T) + \cdots + \mathrm{PER}_{K+L}(\mat T)  \right)
\leq 0,$$
which is a very simple ``positive slope'' detector.
It introduces a small overhead: $K+L$ dimensions are required to
decided if the signal space is of dimension $\leq K$.

Despite its simplicity and its heuristic motivations, this criterion
proved to be robust in practice as shown in \figurename~\ref{fig:perperf}.

\skel{
\begin{itemize}
\item the signal subspace spectrum is gradually uncovered during the lanczos
  process 
\item no knowledge of the spectrum as a whole (signal + noise) $\Rightarrow$
  additional ``bootstrap'' cost for MDL, AIC and others 
\item PER: partial essential rank 
\item estimate the ``essential'' rank of current uncovered signal
  space. 
\item observation: the evolution of the PER has an irregularity when the extra
  dimension added to the signal space is mostly made of noise 
\item monitor the evolution of the PER during the lanczos process:
  $\mathcal O \left( K \right)$ iterations for $N >> K$ (cite xu and
  kailath for convergence proof)
\end{itemize}}
Too many components and/or large modelization error indicate the
channels are not sparse.

\section{Numerical results}
\subsection{accuracy: FRI vs CS vs lowpass}
\subsubsection{Setup}
We use the ``FTW rural'' dataset to compare aforementioned algorithms
on the SCS channel estimation problem. The transmitter is
a mobile single antenna device\footnote{Only the first Tx antenna is used.}. The receiver is a ``base-station'' with $P=8$ receiving
antennas. A total of $251$ frames are transmitted on a $120MHz$ wide
channel with a $2GHz$ carrier frequency. The transmission has a high SNR, so we consider the samples
to be the ground truth (infinite SNR). Various SNR conditions are
simulated with the addition of AWGN to the samples.

The DFT pilots are uniformly laid-out every $D=3$ DFT bin. The estimated
channel is used to demodulate 4-PSK coded data symbols occupying the
left over DFT bins, and the obtained \emph{Symbol Error Rate} is the
quality metric used to benchmark the estimation.

\subsubsection{Interpretation of the results}

\begin{figure*}[!t]
  \centering
  {\footnotesize
      \psfrag{T}[c][c]{a) Average over time (all snapshots)}
    \psfrag{S}[c][c]{b) Average over SNR ($-15$dB to $0$dB)}
    \psfrag{y}[c][c]{SER}
     \psfrag{t}[c][c]{tunnel}
    \psfrag{F}[lc][lc]{FRI-PERK}
    \psfrag{L}[lc][lc]{Lowpass interp.}
    \psfrag{R}[lc][lc]{RA-ORMP}    
    \psfrag{x}[c][c]{Time [s]}
    \psfrag{z}[c][c]{SNR [dB]}
    \includegraphics[width=\textwidth]{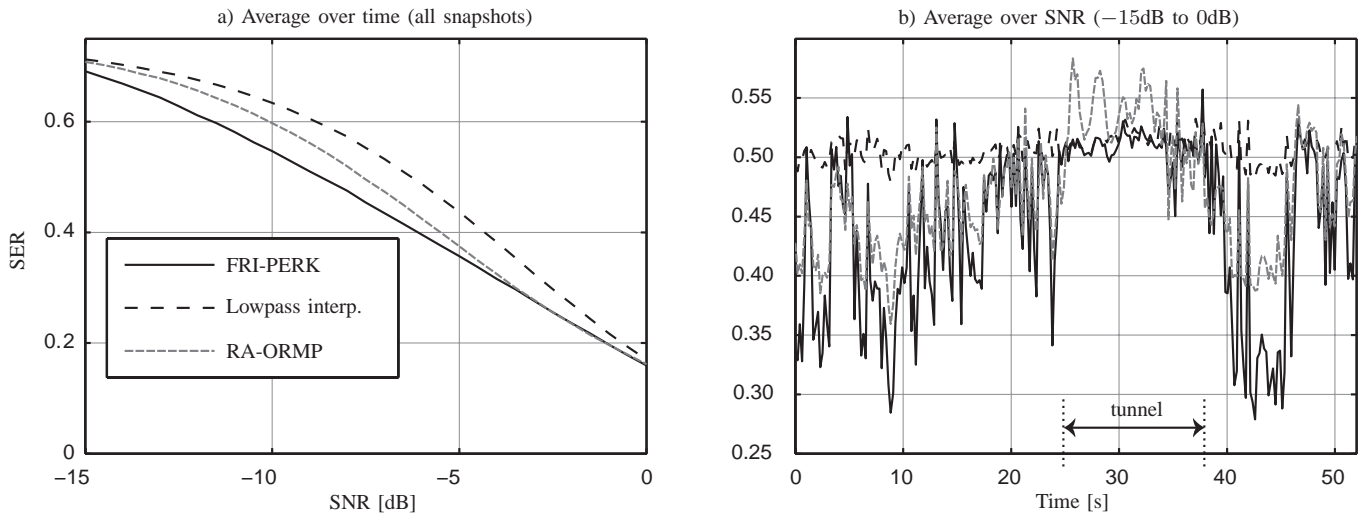}
    }
\caption{From a) we conclude that sparse recovery lowers the SER compared to non-sparse
  recovery below $0$dB of SNR. If we look at the SER over time in b),
  we see that FRI-PERK is robust in the sense that if the input signal
is not sparse, its performs approximately as well as a non-sparse
recovery.}
\label{fig:avgSER}
\end{figure*}

\begin{figure*}[!t]
  \centering
  {\footnotesize
    \psfrag{a}[c][c]{a)}
    \psfrag{b}[c][c]{b)}
    \psfrag{c}[c][c]{c)}
    \psfrag{m}[c][c][1][90]{magnitude}
    \psfrag{t}{Time [sample]}
    \psfrag{g}[c][c]{original CIR}
    \psfrag{n}[c][c]{noisy measurements}
    \psfrag{F}[c][c]{FRI-PERK estimation}
    \psfrag{R}[c][c]{RA-ORMP estimation}
  \includegraphics[width=\textwidth]{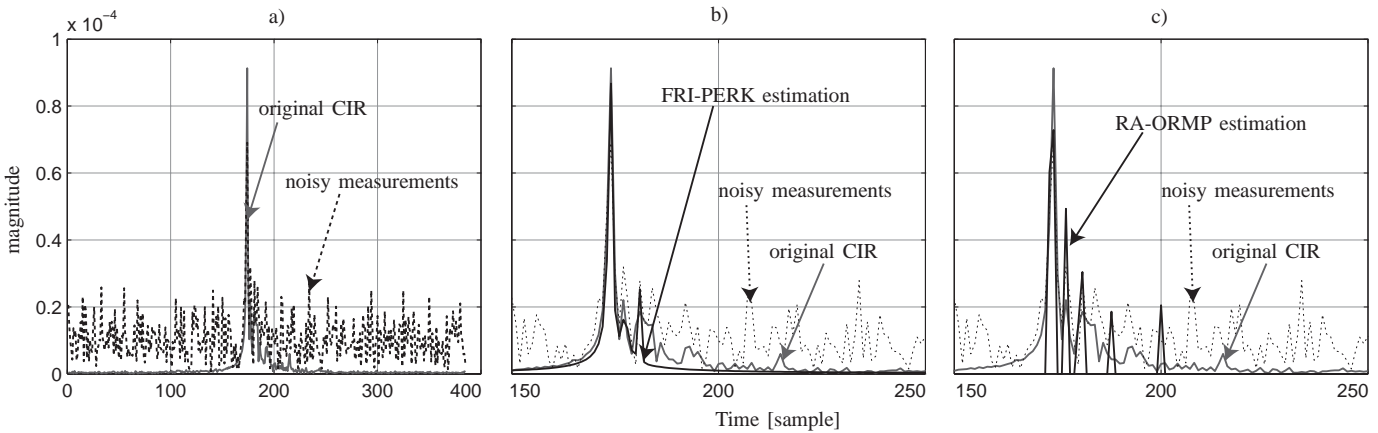}
  }
\caption{This figure compares the estimation result of FRI-PERK and
  RA-ORMP. The input signal is the first frame received at the first
  antenna corrupted with AWGN to obtain $-5$dB of SNR. Panel (a) shows
the measurements and the original CIR. Panel (b) shows a portion of
interest of the CIR estimated with FRI-PERK. The PER criterion
estimates $K=3$, and by visual inspection the three signal components
found match the largest ones of the original signal. Panel (c) shows
the result obtained with RA-ORMP. The discrete sparsity model causes the
estimation to be more sensitive to uncorrelated noise, as spurious
spikes contributed by noise are estimated as signal components.}
\label{fig:F1P1}
\end{figure*}

From \figurename~\ref{fig:avgSER} we may conclude that

\begin{itemize}
\item The channels do not exactly fit the SCS model, therefore the
  modelization error becomes larger than the noise at high SNR
\item The SCS property helps in lowering the symbol error rate at
  medium to low SNR (below 0 dB)
\item The ``sparsity'' model assumed by FRI (few reflections) match
  the field measurements better than the one assumed by CS (few non-0
  coefficients) as seen in \figurename~\ref{fig:F1P1}.
\item Any algorithm exploiting sparsity must be ``\emph{introspective}'',
  i.e. it must detect when sparsity does not occur, and fall-back to a
  more traditional method whenever it happens. It is exemplified by
  the stroll through the tunnel.
  \end{itemize}

\skel{
\begin{itemize}
\item guess sparsity level using the PER for FRI 
\item use REmBo without knowledge of sparsity: run several times,
  adopt consensus. 
\item run on a synthetic channel: FRI is expected to do very well
  since the model matches perfectly 
\item field data simulation: proof by fire!
\end{itemize}}

\subsection{Computational benchmark}

Both FRI-PERK\footnote{Uses the ARPACK library
  \cite{Lehoucq1998}. Since the size of the Krylov subspace dimension
  must be fixed beforehand, it does not has the ability
  to stop before $K_{\max}$ is reached. } and RA-ORMP were implemented in
MATLAB\footnote{v. R2011b for MacOS X,
 running on a 1.8GHz Intel Core I7 processor and enough 1.3 GHz DDR3 memory}.
Because of the decimation in frequency and the limitation of the
delay-spread in time, the projection in RA-ORMP cannot be realized
with simple FFTs, results are reported in  \figurename~\ref{fig:speed}.(c-d).

If this is not the case, we implemented a fast
version of it and reported results in
\figurename~\ref{fig:speed}.(a-b). Note that such a setup does not apply
to most standards \cite{3gpplte,dvbt}.

\begin{figure*}[!t]
  \centering {\footnotesize
    \psfrag{a}[Bc][Bc]{(a)}
    \psfrag{b}[Bc][Bc]{(b)}
    \psfrag{c}[Bc][Bc]{(c)}
    \psfrag{d}[Bc][Bc]{(d)}
    \psfrag{s}[Bc][Bc]{FRI-PER}

    \psfrag{k}[Bc][Bc]{FRI-PERK}
    \psfrag{n}[Bc][Bc]{RA-ORMP}
    \psfrag{o}[Bc][Bc]{fast RA-ORMP}
    \psfrag{P}[Bc][Bc]{$2\times$RA-ORMP}
    \psfrag{O}[Bc][Bc]{$4\times$fast RA-ORMP}
    \psfrag{T}[cc][cc][1][90]{Time [s]}
    \psfrag{C}[cc][cc][1][-90]{\textbf{Contiguous pilots}}
    \psfrag{D}[cc][cc][1][-90]{\textbf{Scattered pilots}}
    \psfrag{B}[cc][cc]{\textbf{Execution time}}
    \psfrag{E}[cc][cc]{\textbf{Equalization gain}}
    \psfrag{S}[cc][cc][1][90]{SNR gain [dB]}
    \psfrag{M}[ct][ct]{$M+1$ (half number of pilots)}

    \includegraphics[width=\textwidth]{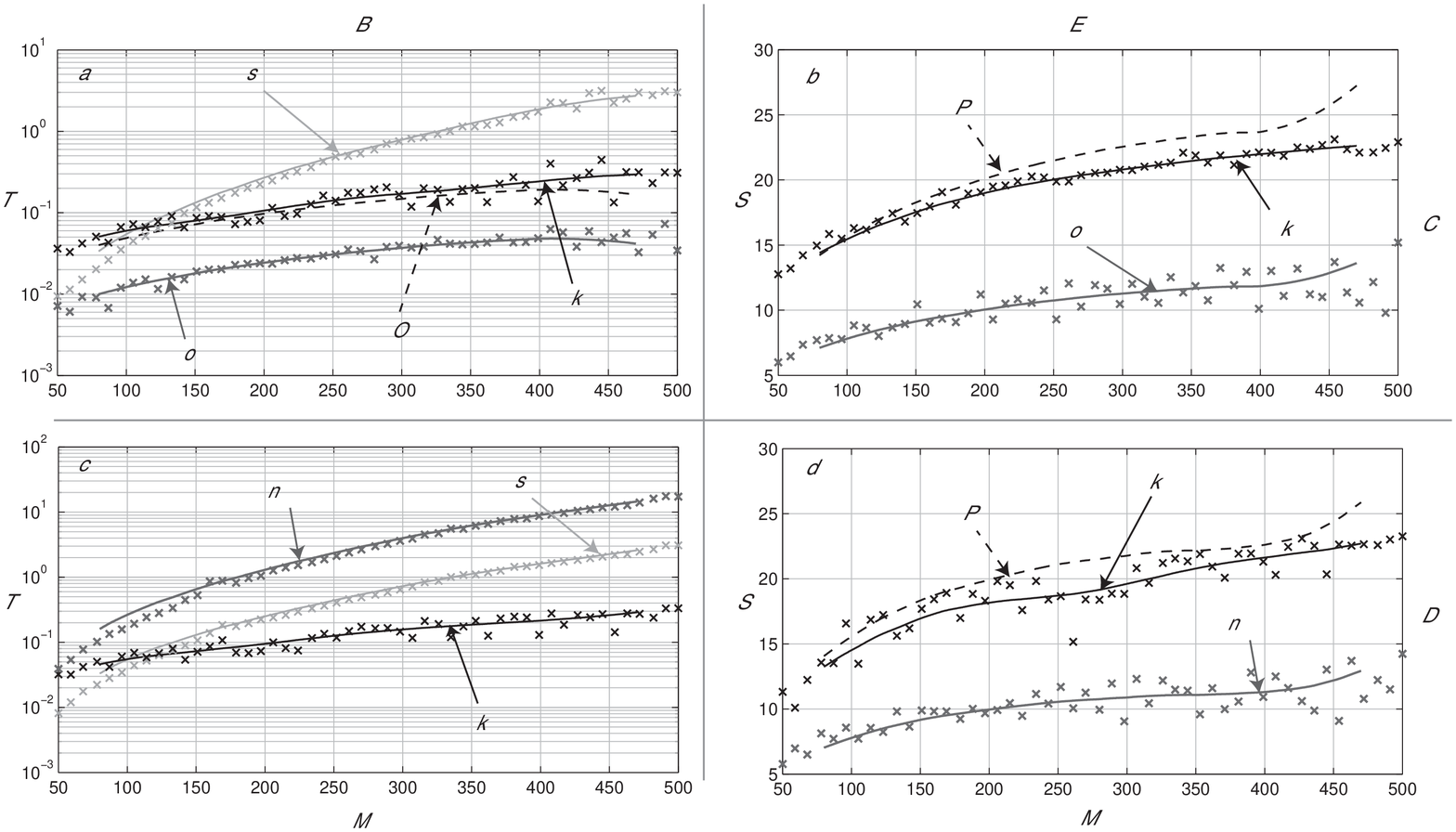}
    }
\caption{A benchmark is run for pilot sequences of length $2M+1=101$
  to $2M+1=1001$. The top row compares FRI-PERK with FRI-PER --- same
  algorithm using a full SVD instead of Krylov subspace projection ---
  and a fast implementation of RA-ORMP using FFTs, which is possible
  since the pilots are contiguous in frequency ($D=1$). The bottom row
follows the same procedure, but with scattered pilots ($D=3$) and a
delay-spread smaller than the frame length. There
is no straightforward fast implementation of RA-ORMP in this case.}
\label{fig:speed}
\vspace{-4mm}
\end{figure*}

Definitive conclusions for hardware implementations cannot be drawn
from these experiments, however we can see that the low asymptotic
complexity of FRI-PERK is not misleading since it provides
improvements for a number of pilots which are encountered in practice
\cite{3gpplte,dvbt} (higher bandwidth modes). For small numbers of
pilots, the direct implementation maybe faster, though we did not used
the tracking capacity of Lanczos algorithm used by choosing an
initial vector lying in the signal space obtained at the previous step
\cite{Xu1994}.

\skel{
\begin{itemize}
\item MATLAB implementation favorizes the plain implementation, but
  some speed improvements.
\item sound algorithmic principles (low-memory, standard dsp
  blocks,...)  $\Rightarrow$ should be fast on embedded hardware.
\end{itemize}}

\section{Conclusion}
In this study, we addressed computational and robustness issues of FRI
techniques applied to channel estimation. The tests conducted on field
measurements show the SCS assumption is relevant at medium to low
SNR where the noise power exceeds the one of model mismatch, and can
be used to substantially lower the SER.

The PER criterion used to estimate the sparsity level gives
satisfactory results, but an upcoming thorough analytical study is required. 

Comparison with discrete techniques from compressed sensing indicate a
trade-off between speed and accuracy compared to FRI-PERK. 
Using a finer discretization (\emph{frames}) allows one to vary this
trade-off, and shall be tested in the future. We believe that with a
proper discretization similar speed/accuracy results shall be met.

\ifCLASSOPTIONcaptionsoff
  \newpage
\fi

\bibliographystyle{IEEEtran}

\bibliography{IEEEabrv,./jetcas}

\vspace{-10mm}
\begin{biography}[{\includegraphics[width=1in,height=1.25in,clip,keepaspectratio]{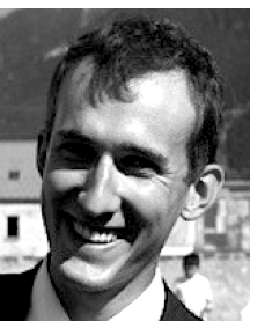}}]{Yann
  Barbotin} was born in Ambilly, France in 1985. He received the B.Sc. and
M.Sc. degrees in Communication Systems from the Swiss Federal Institute
of Technology, Lausanne (EPFL) in 2006 and 2009, respectively. As an
intern  at Qualcomm Inc. in 2009 and an ongoing scientific
collaboration he is the co-inventor of five patents on sparse sampling. He is
now a PhD candidate in the LCAV laboratory at EPFL. His research
interests are in fundamental limitations of estimation problems and sampling.
\end{biography}
\begin{biography}[{\includegraphics[width=1in,height=1.25in,clip,keepaspectratio]{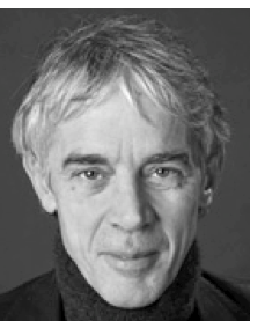}}]{Martin
    Vetterli}
  Martin Vetterli was born in 1957 and grew up near Neuchâtel. He received the Dipl. El.- Ing. degree from Eidgenössische Technische Hochschule (ETHZ), Zurich, in 1981, the Master of Science degree from Stanford University in 1982, and the Doctorat ès Sciences degree from the Ecole Polytechnique Fédérale, Lausanne, in 1986.
After his dissertation, he was an Assistant and then Associate Professor in Electrical Engineering at Columbia University in New York, and in 1993, he became an Associate and then Full Professor at the Department of Electrical Engineering and Computer Sciences at the University of California at Berkeley. In 1995, he joined the EPFL as a Full Professor. He held several positions at EPFL, including Chair of Communication Systems and founding director of the National Competence Center in Research on Mobile Information and Communication systems (NCCR-MICS). From 2004 to 2011 he was Vice President of EPFL and since March 2011, he is the Dean of the School of Computer and Communications Sciences.
He works in the areas of electrical engineering, computer sciences and applied mathematics. His work covers wavelet theory and applications, image and video compression, self-organized communications systems and sensor networks, as well as fast algorithms, and has led to about 150 journals papers. He is the co-author of three textbooks, with J. Kovacevic, "Wavelets and Subband Coding" (Prentice-Hall, 1995), with P. Prandoni, "Signal Processing for Communications", (CRC Press, 2008) and with J. Kovacevic and V. Goyal, of the forthcoming book "Fourier and Wavelet Signal Processing" (2012).
His research resulted also in about two dozen patents that led to technology transfers to high-tech companies and the creation of several start-ups.
His work won him numerous prizes, like best paper awards from EURASIP in 1984 and of the IEEE Signal Processing Society in 1991, 1996 and 2006, the Swiss National Latsis Prize in 1996, the SPIE Presidential award in 1999, the IEEE Signal Processing Technical Achievement Award in 2001 and the IEEE Signal Processing Society Award in 2010. He is a Fellow of IEEE, of ACM and EURASIP, was a member of the Swiss Council on Science and Technology (2000-2004), and is an ISI highly cited researcher in engineering.
\end{biography}

\end{document}